\documentclass{article}
\usepackage{ijcai17}
\usepackage{times}
\usepackage{amsmath,amsthm}
\usepackage{graphicx}
\usepackage{tabulary}
\usepackage{tabularx}
\usepackage{booktabs}
\usepackage{mathtools}  
\usepackage{amssymb}
\usepackage[noend]{algpseudocode}
\usepackage[linesnumbered,ruled,vlined]{algorithm2e}
\usepackage{url}
\newtheorem{theorem}{Theorem}
\usepackage{array}
\usepackage{xcolor}
\usepackage{enumitem}
\setlist{noitemsep,topsep=0pt,parsep=0pt,partopsep=0pt}
\usepackage[small]{caption}

\title{Diverse Weighted Bipartite $b$-Matching}
\author{Faez Ahmed\\ 
Dept. of Mechanical Engg.\\
University of Maryland\\
faez00@umd.edu
\And
John P. Dickerson\\
Dept. of Computer Sci.\\
University of Maryland\\
john@cs.umd.edu
\And
Mark Fuge\\ 
Dept. of Mechanical Engg.\\
University of Maryland\\
fuge@umd.edu
}
\newcommand{\eg}{{\em e.g.}}
\newcommand{\ie}{{\em i.e.}}

\newcommand{\etal}{{\em et~al.}}
\newcommand{\cutequationup}{\vspace*{-0.12in}}
\newcommand{\PoD}{PoD}

\begin{document}

\maketitle

\begin{abstract}
Bipartite matching, where agents on one side of a market are matched to agents or items on the other, is a classical problem in computer science and economics, with widespread application in healthcare, education, advertising, and general resource allocation. A practitioner's goal is typically to maximize a matching market's economic efficiency, possibly subject to some fairness requirements that promote equal access to resources.  A natural balancing act exists between fairness and efficiency in matching markets, and has been the subject of much research.

In this paper, we study a complementary goal---balancing \emph{diversity} and efficiency---in a generalization of bipartite matching where agents on one side of the market can be matched to \emph{sets} of agents on the other.  Adapting a classical definition of the diversity of a set, we propose a quadratic programming-based approach to solving a supermodular minimization problem that balances diversity and total weight of the solution.  We also provide a scalable greedy algorithm with theoretical performance bounds.  We then define the \emph{price of diversity}, a measure of the efficiency loss due to enforcing diversity, and give a worst-case theoretical bound.  Finally, we demonstrate the efficacy of our methods on three real-world datasets, and show that the price of diversity is not bad in practice. Our code is publicly accessible for further research.\footnote{\url{https://github.com/faezahmed/diverse_matching}} 
 
\end{abstract}

\section{Introduction}\label{sec:intro}
Bipartite matching problems pair an agent or item on one side of a market to an agent or item on the other. Weighted bipartite $b$-matching generalizes this problem to the setting where matches have a real-valued quality, and agents on one side of the market can be matched to a cardinality-constrained \emph{set} of items or agents on the other side; real-world examples include matching children to schools~\cite{Kurata15:Controlled,Drummond15:SAT}, reviewers to manuscripts~\cite{Charlin13:Toronto,Liu14:Robust}, donor organs to patients~\cite{Bertsimas13:Fairness,Dickerson15:FutureMatch}, and workers to firms~\cite{Horton17:Effects}.

Often, a matching market's central goal is to maximize economic efficiency subject to some fairness constraints, such as ensuring equal opportunity amongst participants.  For example, a firm might wish to maximize the number of open positions filled subject to a fairness constraint: a firm must interview a representative number of workers from marginalized backgrounds.  Yet, the firm also cares about the entire cohort's ability: the workers it hires should hold high quality, yet complementary skill sets. This paper studies the trade-off between economic efficiency and \textit{diversity}, where matchings provide good coverage over different \emph{classes} of item or agent.

A representative example this paper considers is matching academic papers to possible reviewers. A paper might have highest relevance to three reviewers who all come from the same lab group, perhaps because they all published heavily in a similar area. Existing weighted bipartite $b$-matching (WBM) algorithms would likely assign those three reviewers to the same paper. Is this outcome desirable? On the one hand, yes, because they have expertise related to the paper. On the other hand, those reviewers would stress similar points, given their common background. So the paper may only improve in a narrow (albeit important) direction. What if we wanted to diversify the reviewer backgrounds\textemdash to find reviewers well-suited to the paper \textit{and} complementary to each other? Ideally, the reviews would remain high quality, but would cover different, complementary aspects of the paper.

This paper addresses how to compute diverse matchings under various constraints, bounds the loss on economic efficiency due to using a diverse matching, and shows in simulation and on data from three real-world bipartite matching problems that it is possible to achieve diverse matchings with limited cost to economic efficiency.

\subsection{Related Work}\label{sec:intro-rw}

In practice, the weighted bipartite $b$-matching (WBM) problem\textemdash find the feasible matching with maximum weight\textemdash has arisen naturally as a problem in many fields, such as: protein structure alignment~\cite{Krissinel04:Secondary}; computer vision~\cite{Belongie02:Shape}; estimating text similarity~\cite{Pang16:Text}; VLSI design~\cite{Huang91:Data}; and matching reviewers to papers in peer-review systems~\cite{Charlin13:Toronto,Liu14:Robust,Tang10:Expertise}. Driven by practical application, such previous work aims to maximize economic efficiency.  We will compare against this objective in the present work.

This paper incorporates diversity objectives into the WBM problem. The closest related paper, due to Liu \emph{et al.}~\shortcite{Liu14:Robust}, performs a node-specific diversity-inspired preprocess before solving a related matching problem; in our work, we consider the ``global'' diversity of the full matching, a function of the diversity of \emph{sets} of vertices.  Other papers have addressed diversity in ranking problems (\eg, diverse recommendations~\cite{Adomavicius12:Improving,sha2016framework,ashkan2015optimal}), but not for matching. Past approaches mathematically represent \emph{coverage} of a set of items, such that a diverse set better \emph{covers} the space of items. This coverage is often defined via diminishing marginal utility over a space, such that adding more items to nearby areas of space is less useful. There are many application-dependent choices for what such a space entails including vector spaces such as text vectors~\cite{Puthiya16:Coverage} or metrics over graphs~\cite{Zhang05:Improving}, among others. To represent diminishing marginal utility, families of \emph{submodular functions} are natural candidates that have shown promise in diversity tasks like document summarization~\cite{Lin11:Class}. We will use similar reasoning when defining our objectives.

The most similar work to ours is due to Chen~\emph{et al.}~\shortcite{Chen16:Conflict}, who propose Conflict-Aware WBM (CA-WBM). They consider conflict constraints between vertices on the same side of a bipartite graph. In CA-WBM, if two vertices are in conflict, they may not both be matched to a vertex on the other side of the graph. CA-WBM enforces a kind of binary diversity by manually defining conflicts between specific nodes. In contrast, this paper treats diversity as an objective, not a constraint, allowing us to flexibly control the degree to which a matching algorithm encourages or discourages diverse solutions to the standard WBM problem. This is useful when one wants conflicts or diversity to vary in degree, or trade off diversity with other measures of match quality.

\subsection{Our Contributions}\label{sec:intro-contributions}
This work studies the trade-off between diversity and efficiency in matching markets. This is different from earlier work as the diversity measure is modeled as an objective and not as constraints, and diversity is defined over \emph{sets} of items. 

Our main contributions are as follows.
\begin{enumerate}
\item We formulate the diverse weighted bipartite $b$-matching optimization problem.
\item We propose a polynomial-time greedy algorithm for constrained $b$-matching, and prove performance bounds on that relative to the NP-hard main problem.
\item We show via simulation and data from three large real-world bipartite matching problems that our method produces matchings with much higher diversity than standard efficient matchings, at little overall cost to economic efficiency.
\end{enumerate}

In the following Section~\ref{sec:prelims}, we formalize the weighed $b$-matching optimization problem; then, in Section~\ref{sec:diversity}, we define our diversity-promoting objective and present the \emph{price of diversity}, a measure of the tradeoff in economic efficiency under a diverse matching objective.  Section~\ref{sec:algos} presents an optimal method for solving our problem, a scalable polynomial-time greedy algorithm with performance bounds, and a worst-case bound on the price of diversity.  Section~\ref{sec:experiments} shows via simulation and on real data from three matching problems that (i) our method promotes diversity in matching, (iii) the greedy approximate algorithm is both scalable and performs comparable to optimal, (ii) both algorithms retain dramatically more efficiency than our worst-case bounds implied; that is, the price of diversity in practice is quite good.

\section{Weighted Bipartite Matching}\label{sec:prelims}

\begin{figure}
\centering
\includegraphics[width=0.5\columnwidth]{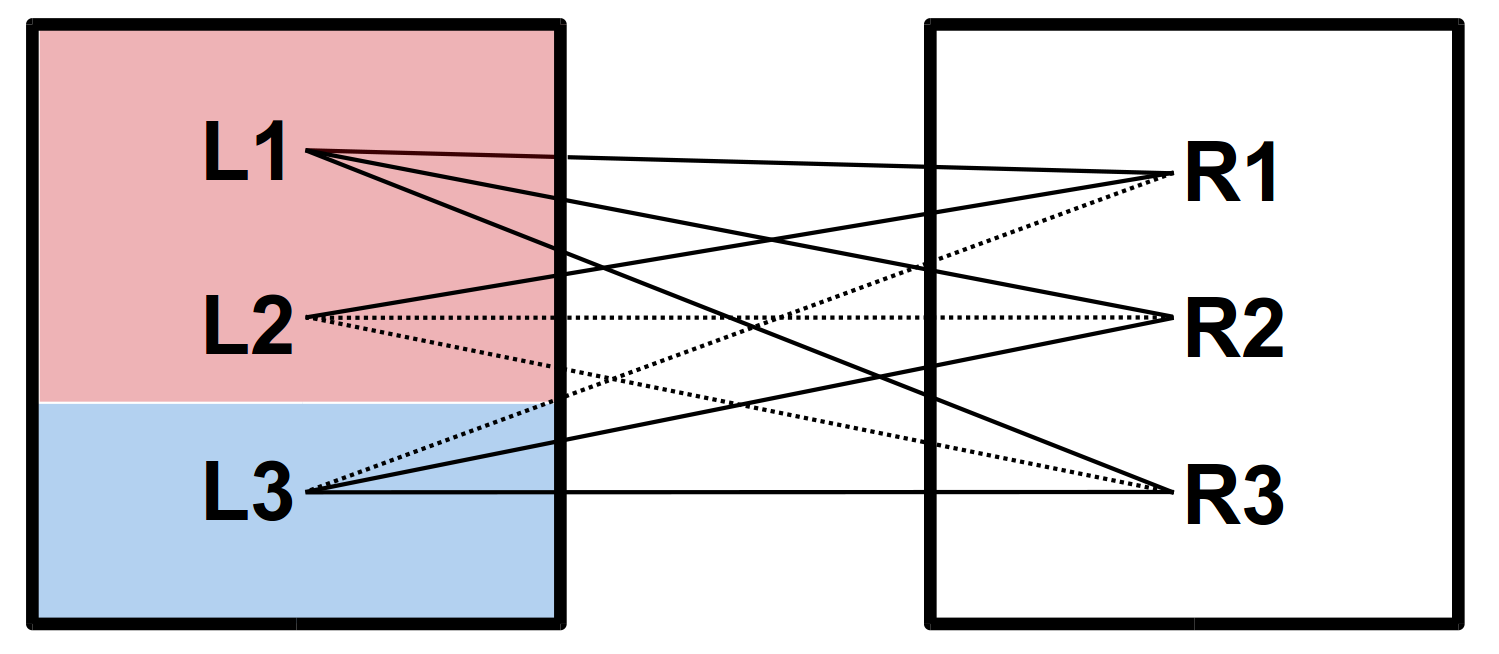}
\caption{Bipartite $b$-matching problem where the left side nodes are divided into two clusters.} 
\label{fig:match}
\end{figure}

Weighted bipartite $b$-matching is a combinatorial optimization problem formulated as follows. Given a weighted bipartite graph $G=(U,V,E)$ with weights 
$ W:E \rightarrow  R^+ $, where $U$, $V$ and $E$ represent left vertices, right vertices and edges, the weighted bipartite $b$-matching problem is to find a subgraph $T \subset G$ such that each vertex $i$ in $T$ has at most $b$ edges (\ie, a degree constraint). WBM maximizes or minimizes the objective depending on the application.

We use similar notation to Chen \emph{et al.}~\shortcite{Chen16:Conflict} to define a weighed bipartite $b$-matching problem, with two notable differences.  First, we define it as a minimization problem, and second, we define a harder problem which has both node-specific upper- and lower-cardinality constraints.  The constrained weighed bipartite $b$-matching (WBM) problem can be expressed as follows.

\vspace{-0.4cm}
{\small
\begin{equation}
\begin{array}{lrr}
\underset{X}{\text{min}} & f_1 = WX & \\
\text{s.t.}
& L^- \leq AX_i\leq L^+ & \forall i \in \{1, \ldots, M\}\\
& R^- \leq AX_i\leq R^+ &\forall i \in \{M+1, \ldots, M+N\}\\
& x_{ij} \in \{0,1\} & \forall i,j, 1 \le i \le M, 1 \le j \le N 
\end{array}
\label{eq:wbm}
\end{equation}
}
\vspace{-0.2cm}

We have $N$ items on the right side with $R^-$ and $R^+$ integral lower and upper cardinality constraints, respectively, and $M$ items on the left side with $L^-$ and $L^+$ as integral cardinality bounds.
Here, $X$ is a column vector of binary variables of size $MN$, with $x_{ij}=1$ if left item $i$ is matched to right item $j$, and $x_{ij}=0$ otherwise. $W$ is a matrix of weights $w_{ij}$ representing the local quality of matching items $i$ and $j$.

Items on the same side of bipartite graph cannot be matched; thus, we use $A$ as a linking matrix, such that any row $i$ indicates which nodes are allowed to be connected to item $i$.  We index edges $(i,j)$ uniquely using a function $\ell : E \to \{1, \ldots, MN\}$.  Then, $A$ is an $(M+N) \times MN$ matrix, where $a_{i\ell((i,j))}=1$ if edge $(i,j)$ exists; otherwise, $a_{i\ell((i,j))}=0$.  The degree constraints for left nodes are given by $L^- \leq AX_i \leq L^+$, where $AX_i$ denotes the $i^{\mathit{th}}$ element in (vector) $AX$.  $L^+$ is the upper bound on cardinality of $i^{\mathit{th}}$ node and $L^-$ is the lower bound.

The above formulation shows a constrained matching problem, where nodes on both sides have capacity constraints. This discrete linear optimization problem is NP-hard~\cite{Chen16:Conflict}. Its optimal solution will minimize the weights, emphasizing on efficiency and neglecting diversity.

\section{Diversity in Matching}\label{sec:diversity}
Diversity in matching can be defined as the need to match a node with other nodes from different groups. To add diversity, we consider a scenario where left-side nodes are divided into $K$ groups. Let us say that we want a matching which matches each node on the right side to nodes from different clusters. 
The diversity is calculated using supermodular functions. These functions have been widely used in extractive document summarization~\cite{Lin11:Class} to get a diverse high quality summary of documents. We use a quadratic function which can incorporate diversity by balancing the number of nodes (e.g., items or agents) selected from different clusters.

Let $E_l = \{(1,l), \ldots, (M,l)\}$ be the set of all $M$ edges from a node $l \in \{1,\ldots,N\} $ on the right side of the graph. Let the subset $S_l \subseteq E_l = \{(1,l),...,(m,l)\}$ be the matched $m$ edges for node $l$. Let $(P_i)^l$, $i \in \{1,\ldots,K\}$ is a partition of the ground set $E_l$ into $K$ separate clusters (\ie, $\cup _i P_i^l = E_l$, and $\cap_iP_i^l = \varnothing$). That is, a left item can only belong to one cluster. The weight of an edge from left node $n$ to right node $l$ is $w_{n,l}$. We define the quality of match for node $l$ on the right side as:
\cutequationup
\begin{equation}
f(l) = \sum_{k=1}^K \left \{ {\sum_{j\in S_l\cap P_k^l}w_{j,l}} \right \}^2
\label{eq:div}
\end{equation}

This quadratic function gives lower cost to solutions with even coverage over all clusters. As an example, Figure~\ref{fig:match} shows three nodes on either side, each requiring two edges. If all edge weights $w$ are one, the node-specific utility of a matching $\{(L1,R1), (L2,R1)\}$ is $4$, while the utility of alternate matching $\{(L1,R1), (L3,R1)\}$ is $2$. Hence, by minimizing the function in Equation~\ref{eq:div}, diversity is encouraged. By transforming the matching problem to quadratic minimization, the resultant objective function simultaneously optimizes quality and diversity. In the next section, we provide a formal framework to generalize this function to constrained $b$-matching problems.

To the best of our knowledge, no known general measure exists to measure the performance of diverse $b$-matching methods. Even verification of diverse matching is difficult, due to different definitions of diversity in the literature. One way of comparing our diversity results is to look at the Shannon entropy of a match for each item, with and without our method. Shannon entropy has been
used to incorporate diversity in recommendations \cite{qin2013promoting,DiNoia:2014:AUP:2645710.2645774} and also
widely used in the ecological literature as a diversity index. It
quantifies the uncertainty in predicting the cluster label of an individual that is taken at random from the dataset. Here entropy of a node is given by: $-\sum_{i=1}^K{(p_k \log{p_k})}$, where $p_k$ is the proportion of selected edges in cluster $K$.

Entropy for an item is maximized if it is matched to other items with even coverage of different clusters; it is zero when all such items are from the same cluster. 
Hence, the impact of diverse matching can be measured as improvement in average entropy. We define the \emph{entropy gain} (EG) as:
\begin{equation}
EG = \frac{\mbox{Average entropy using a diverse matching rule}}{\mbox{Average entropy using WBM}}
\end{equation}

We also propose a new metric to measure the efficiency lost due to diversity. We define the \emph{price of diversity} (PoD) as:
\begin{equation}
PoD=\frac{\mbox{Utility using WBM}}{\mbox{Utility using a diverse matching rule}}.
\end{equation}

Later in the paper, we will show in simulation and on real data that the entropy gain under our proposed diverse matching method is high, at very little cost to overall efficiency.

\section{Exact and Approximate Algorithms}\label{sec:algos}

In our bipartite matching formulation with utility minimization, the degree constraints $L^-$ and $R^-$ can be interpreted as setting the demand. The short side of market determines the number of edges in the matching, which is $\min\{ML^-, NR^-\}$. If the right side is short, the maximum capacity on the left should be more than the demand: $NR^-  \le ML^+$ for any matching to be feasible. For the purpose of this paper, we always assume that right side is the short side of market and the left side is clustered into groups. The cardinality constraints make the problem more difficult than what has usually been solved for matching as nodes cannot be matched independent of each other.

\subsection{Diverse WBM}

To generalize the quadratic function (cf. Equation~\ref{eq:div}) to an optimization framework for all nodes, we define a $MN \times MN$ block-diagonal matrix $B=\text{diag}(B_1,\ldots, B_M)$ such that:
\[
  B_l=\begin{cases}
               w_i \cdot w_j  ~~\text{if edges } i,j\in~P_k^l~\text{(same cluster)}\\
               0~~\text{otherwise}
            \end{cases}
\]

$B_l$ is the block diagonal matrix for every right node, with $K$ blocks on the diagonal corresponding to each cluster. Matrix $B$ is a diagonal matrix for all $B_l$ matrices combined. Later, we show a visualization of the symmetric $B_l$ matrix in Fig.~\ref{fig:block} for five clusters for reviewer matching application. Hence, the optimization problem for Diverse WBM (D-WBM) can be written as:

\begin{equation}
\begin{aligned}
& \underset{X}{\text{minimize}}
& &f_2(X) =  X^TBX \\
\end{aligned}
\label{eqdpwbm}
\end{equation}


To show that this formulation is equivalent to Eq.~\ref{eq:div}, let us again consider $R1$ in Fig.~\ref{fig:match} with two clusters, three edges and unit weights.
Using Eq.~\ref{eqdpwbm} for the node-specific utility of a matching $\{(L1,R1), (L2,R1)\}$  is $[1, 1, 0]’[1, 1, 0; 1, 1, 0; 0, 0, 1][1, 1, 0]~=~4$ and the utility of alternate matching $\{(L1,R1), (L3,R1)\}$ is $[1, 0, 1]’[1, 1, 0; 1, 1, 0; 0, 0, 1][1, 0, 1]=2$. This is same as obtained by Eq.~\ref{eq:div} before. The constraints and variables are the same as in WBM (cf. Equation~\ref{eq:wbm}). Our new model has a quadratic objective with linear constraints and integrality requirement for variables. 
We solve it using two different approaches, first using Gurobi's Mixed Integer Quadratic Programming (MIQP) Solver~\cite{gurobi}, and second by using a novel greedy algorithm that builds up a set by minimizing marginal gain.
Next, we propose this greedy algorithm and give bounds on its performance.

\subsection{Greedy Diverse WBM}
The objective of the D-WBM formulation is supermodular, that is, adding new elements leads to (strictly) increasing differences. Hence a method which greedily adds edges by minimizing the marginal gain can attain reasonable performance bounds~\cite{Nemhauser78:Analysis}.
Secondly, solving the MIQP exactly requires storing the block diagonal matrix; for large problems, MIQP may run out of memory as the number of non-zero terms in the matrix scales by $M^2N$.  

We propose an algorithm which incrementally \emph{satisfies} the lower degree constraints for all nodes. It does this by increasing the lower bound of all nodes unit step at a time and selecting edges by minimizing marginal gain in the objective $f_2$. In selecting edges, it gives preference to the set of nodes with unsatisfied lower bound. This ensures that the greedy selection always finds a feasible matching with good empirical performance.

\begin{algorithm}
\DontPrintSemicolon 
\KwIn{A set of $N+M$ nodes, bounds $L^-, L^+, R^-, R^+$ and edge weights matrix $B$}
\KwOut{A feasible matching}
$C \gets \emptyset$\;
\For{$i \gets 1$ \textbf{to} $\max\{L^-, R^-\}$}{
$L_i^- = \min\{i, L^-\}; R_i^- = \min\{i, R^-\}$\;
\For{$j \gets 1$ \textbf{to} $N+M$}{
\If{$j$'s lower bound is not satisfied} {
Select the feasible edge $e$ with lowest marginal gain $f(C \cup \{ e \}) - f(C)$ whose opposite node's lower bound is not satisfied\\
If no such node exists, pick the feasible edge with lowest marginal gain.
  }
 }
}
\Return{$C$}\;
\caption{{\sc GD-WBM} Greedy Diverse Matching}
\label{algo:change}
\end{algorithm}

For a right constrained matching, the matching for every right node is independent of others as they do not have overlapping constraints. Hence GD-WBM achieves a $(1- 1/e)$-approximation of the optimum due to its supermodular objective function. In practice, Section~\ref{sec:experiments} will show that GD-WBM performs much better than the lower bound.

\subsection{Price of Diversity Bound}
In this section, we propose lower bounds on the price of diversity (\PoD{})---the utility loss due to diverse matching---in right-constrained market.  Theorem~\ref{thm:pod} gives such a bound.

\begin{theorem}\label{thm:pod}
The worst-case Price of Diversity (\PoD{}) for a right-constrained diverse matching is: 
\begin{equation}\label{eq:thm-pod}
\begin{aligned}
\frac{1}{N} \sum_{l=1}^{N}{\frac{z_l}{1+\sqrt{R^--1}\sqrt{{z_l}^2-1}}}, \text{ where } z_l=\frac{\sum{w_{jl}}}{\min{w_{jl}}}
\end{aligned}
\end{equation}
\end{theorem}
\begin{proof}[Proof sketch.]
We wish to find a problem instance where the best diverse matching has high weight under the utilitarian objective.  Such \PoD{} for a matching will be minimized when diversity leads to all low-weight weight edges being replaced by higher weight edges. Such a situation can occur when WBM provides a match for a right node with all $m$ edges going to left nodes in the same cluster $1$. Let $\{w_1, \ldots, w_m\}$ be such edge weights. In this instance, let D-WBM select edges going to $m$ unique clusters. Here, D-WBM will select the single edge with least weight from each cluster, including cluster $1$. Call those edges $\{a_1, \ldots, a_m\}$ be the selected edge weights by D-WBM. The edge weights will satisfy the following constraints: 
\begin{equation}
\begin{aligned}
&\sum_{i=1}^m{w_i} \le \sum_{i=1}^m{a_i} ~~~~\text{and}~~~~ (\sum_{i=1}^m{w_i})^2 \ge \sum_{i=1}^m{{a_i}^2}
\end{aligned}
\end{equation}

Both WBM and D-WBM select edge $a_1=\min_{i\in[m]}{w_i}$ from cluster 1. To minimize \PoD{}, the denominator---$f_1$ of the diverse matching---$\sum_{i=1}^m{a_i}$ should be maximized. Using the Lagrangian method, this constrained maximization problem is solved, with optima occurring at the surface of a hypersphere and $a_2=a_3 \ldots =a_m$.
\end{proof}

If the minimum weight for every right node is $0$, by taking limits on Eq.~\ref{eq:thm-pod}, the \PoD{} becomes $\frac{1}{\sqrt{R^--1}}$.
In the succeeding section, we will show that Theorem~\ref{thm:pod} is quite conservative.

\section{Results and Discussion}\label{sec:experiments}

We begin by comparing the two methods' \PoD{} to its theoretical bound on a synthetic dataset.  Next, we solve the $b$-matching problem on three datasets, one for movie recommendation and two for papers--reviewers matching. 
We analyze the effect of problem size by increasing the number of nodes and the cardinality requirements. 
Finally, we discuss the trade-off between diversity and utility.

\subsection{Artificial Dataset}
In this section, we simulate matching $10$ nodes with another $10$ nodes; by Theorem~\ref{thm:pod}, the worst-case \PoD{} is $0.5$. Weights are selected randomly from a uniform distribution between $0$ to $1$ and the cluster labels of the left-side nodes are selected randomly. For right-constrained matching, we use $R^-=5$, implying that each right side node will be matched to at least five items. The number of clusters are varied from $2$ to $10$, and $100$ trials are done to compare D-WBM with WBM. 

Figure~\ref{fig:artificial} shows that in practice, \PoD{} is never below $0.9$ despite random clusters and weights. Also, EG generally decreases when \PoD{} increases. The greedy approximation GD-WBM finds the same matches as D-WBM for all cases.

\begin{figure}
\centering
\includegraphics[width=0.95\columnwidth]{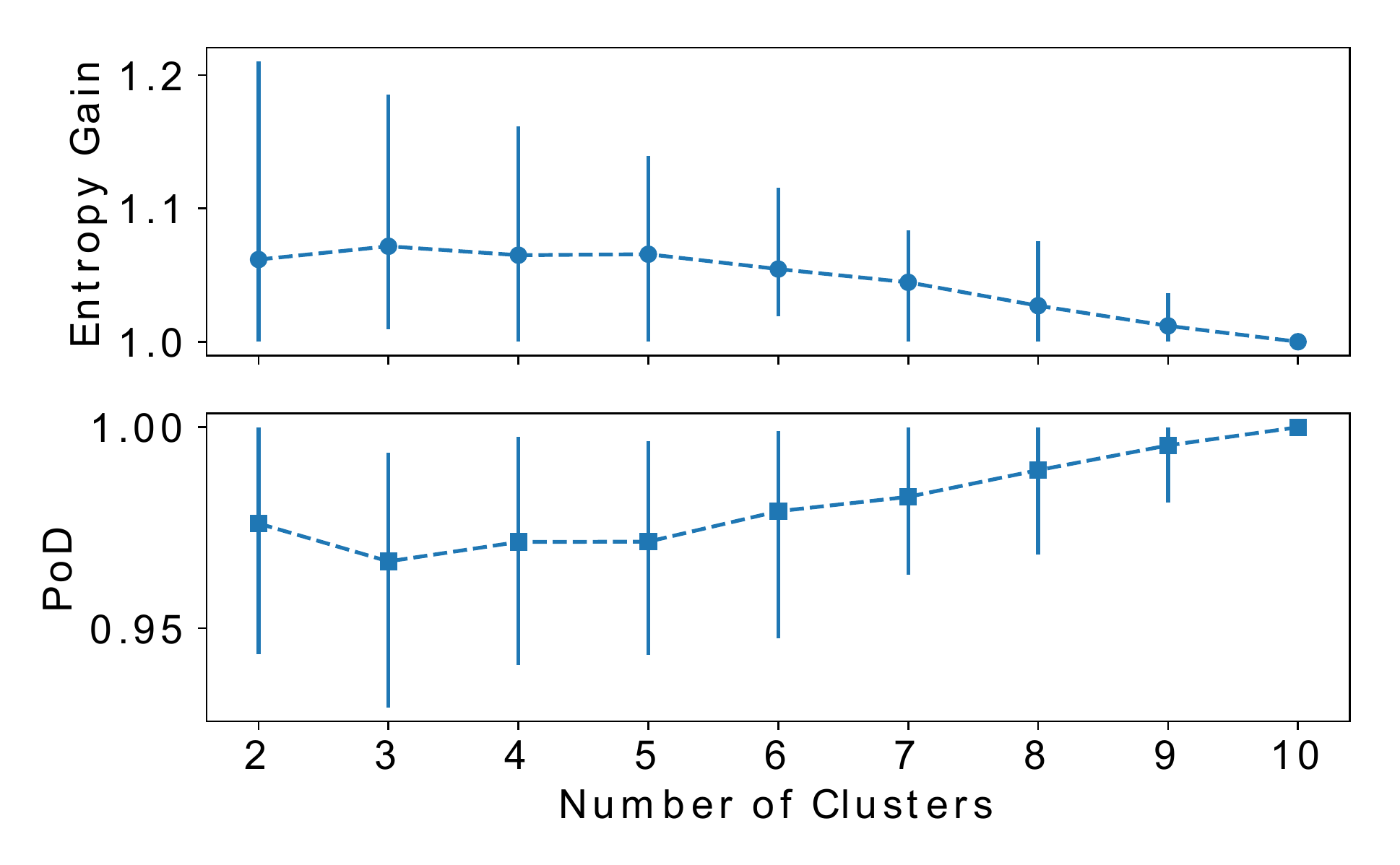}
\caption{\PoD{} and EG for a simulated dataset showing the average \PoD{} with $5^{\mathit{th}}$ and $95^{\mathit{th}}$ percentile values. D-WBM unilaterally outperforms the worst-case \PoD{} of $0.5$.}
\label{fig:artificial}
\end{figure}

\subsection{Application to MovieLens dataset}

This example considers matching movies to users, while ensuring that the movies contain diverse genres. We use a subset of the MovieLens 1M dataset \cite{harper2016movielens}, which includes one million ratings by 6,040 users for 3,900 items. This dataset contains both users' movie ratings between 1 and 5 and genre categories for each movie (\eg, comedy, romance, and action).

We first train a standard collaborative recommender system \cite{reco} to obtain ratings for all movies by every user. We cluster the movies into $5$ clusters using their vector of $18$ genres using spectral clustering,\footnote{The exact choice of recommender system and clustering algorithm is not central to paper; it just helps set up the graph.} so that each movie gets a unique cluster label. We solve the right constrained matching problem for $500$ movies and make recommendations for $500$ users with at least $10$ recommendations for every user. Table~\ref{tab:pred2} shows that with average EG of $1.45$, D-WBM finds a more diverse matching for all users and on average a user loses only $1\%$ utility for this gain. To save computational time, in all our simulations we terminate D-WBM after 1 hour and take the best solution, while WBM converges to true optima. Hence the results are conservative estimates.

To further understand the matching result, we compare the movie recommendations by D-WBM and WBM for User ID 423. WBM matches her to movies that are all either Comedies or Dramas, with an average predicted rating of $4.07$. In contrast, D-WBM matches the user with movies from five different clusters, with an average movie rating of $4.05$, showing negligible loss in efficiency. Table \ref{tab:movielens} compares the recommended genres. While we chose to promote diversity in genre, the MovieLens dataset also provides information about the user's gender, age group, and occupation; D-WBM can encourage other types of diversity in movie recommendation.

\begin{table}
\centering
\begin{tabular}{*6c}
\hline
\multicolumn{2}{c}{Dataset}  &  \multicolumn{2}{c}{D-WBM} & \multicolumn{2}{c}{GD-WBM}\\
\hline
{}   &   {} & \PoD{}   & EG    & \PoD{}   & EG\\
\multicolumn{2}{c}{Movie-Lens}  &  0.99 & 1.45   & 0.99  & 1.45\\
\multicolumn{2}{c}{UIUC Reviewer}  &  0.92 & 1.63   & 0.83  & 1.60\\
\multicolumn{2}{c}{Sugiyama} &  0.94 & 4.28   & 0.93  & 4.28\\
\hline
\end{tabular}
\caption{Price of Diversity and Entropy Gain results for three real world datasets.}
\label{tab:pred2}
\end{table}

\begin{table}[!htbp]
\centering
\begin{scriptsize}
\begin{tabular}{|p{0.8cm}|p{2.4cm}|p{0.8cm}|p{2.8cm}|}
\multicolumn{2}{c}{D-WBM} & \multicolumn{2}{c}{WBM}\\
\hline
Cluster  & Genres   &Cluster   & Genres\\
\hline
 2 &  Drama, Thriller   & 2  & Drama,Thriller\\
1 & Adventure,Sci-Fi   & 2  &Crime,Drama,Thriller\\
3 & Documentary  & 0  & Comedy\\
3 &  Documentary   & 0  & Comedy,Drama\\
 0 & Comedy,Romance   & 0  & Comedy,Romance\\
4 & Drama,Horror  & 2  & Drama\\
4 &  Horror,Sci-Fi,Thriller   & 2  & Drama\\
2 & Drama   & 2  & Drama\\
0 & Comedy,Mystery  & 0  & Comedy,Mystery\\
1 &   Action,Thriller   & 2  & Action,Crime,Drama\\
\hline
\end{tabular}
\end{scriptsize}
\caption{Genres and Cluster labels of ten movies recommended to a user by WBM and D-WBM. The movies by D-WBM provide a broader genre coverage.}
\label{tab:movielens}

\end{table}

\subsection{Application to Reviewer Assignment}
In this section, we present an application of diverse matching to automatically determine the most appropriate reviewers for a manuscript by also ensuring that reviewers are different from each other. 

\subsubsection{UIUC Multi-Aspect Review Assignment Dataset}
We use the multi-aspect review assignment evaluation dataset \cite{karimzadehgan2009constrained} which is a benchmark dataset from UIUC. It contains $73$ papers accepted by SIGIR $2007$, and $189$ prospective reviewers who had published in the main information retrieval conferences. The dataset provides $25$ major topics and for each paper in the set, an expert provided $25$-dimensional label on that paper based on a set of defined topics. Similarly for the $189$ reviewers, a $25$-dimensional expertise representation is provided.

For the reviewer-paper bipartite graph, edge weights between each test paper and reviewer are set as the cosine distance of their label vectors. We cluster the reviewers into $5$ clusters based on their topic vectors using spectral clustering.
We set the constraints such that each paper matches with at least $3$ reviewers and every reviewer is allocated at least $1$ paper, while no reviewer is allocated more than $10$ papers.

Despite the constraints, D-WBM finds a diverse matching with \PoD{} of $0.92$. GD-WBM provides an average entropy gain of $1.60$ but pays a higher price of diversity as shown in Table~\ref{tab:pred2}. To delve deeper into the results, we take the example of $48^{\mathit{th}}$ paper titled ``Towards musical query-by-semantic-description using the CAL500 data set'' from the dataset. This paper is labeled as related to Topics T8 (Multimedia IR), T16 (Language models) and T3 (Other machine learning). WBM matches it to three reviewers with IDs $43$, $34$, and $158$\footnote{More information on the reviewers is available here:  \url{http://sifaka.cs.uiuc.edu/ir/data/review.html}}.  If we analyze their topic vectors, we find that all of them have the Language Models (T16) topic common between them. Not surprisingly, they are all found by our clustering method to be in the same cluster, as shown in Fig.~\ref{fig:block}.

On the other hand, diverse matching provides a match with three reviewers (IDs $131$, $153$ and $158$) from three different clusters. Reviewer $131$ provides a balance between IR and Language Model topics but also works on User Studies. Reviewer $153$ works on many topics relevant to the paper \textemdash Multimedia IR (T8), ML (T2, T3) and Text Categorization (T1). In D-WBM's reviewer set, Reviewers $153$ and $131$ have no common aspect between them while Reviewer $158$ shares interests with both. Having a set of reviewers who are similar to the query paper but complementary in skillsets may provide a well-rounded review with good coverage of different viewpoints. GD-WBM also finds a match which has higher entropy (three different clusters) than WBM. 
\begin{figure}
\centering
\includegraphics[width=0.8\columnwidth]{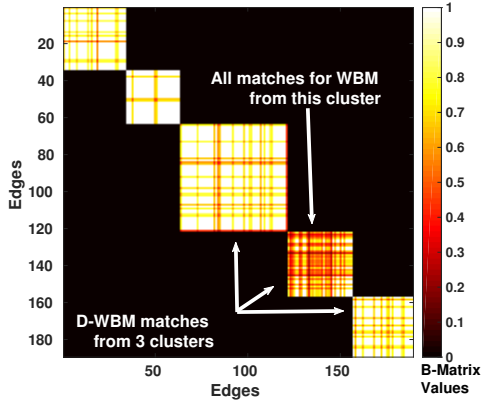}
\caption{Block Diagonal B-Matrix for Paper 48; WBM selects all matches from a single cluster}
\label{fig:block}
\end{figure}

\subsubsection{Scholarly Paper Recommendation Dataset}

To further test our method on matching reviewers with papers, we use the  Scholarly Paper Recommendation dataset provided by Sugiyama \etal~\cite{sugiyama2010scholarly}, which contains $50$ researchers and $100,351$ candidate papers from proceedings in the ACM Digital Library.

We select all papers from KDD 2000 to KDD 2010 from this dataset (a total of $1184$ papers) and find three reviewers for each of them from the given set of 50 reviewers.  We calculate edge weights between papers and reviewers as the cosine distance between the tf-idf vector of the query paper and reviewer's latest paper. No limit of maximum number of papers that a reviewer can review is imposed but each reviewer must be allocated at least one paper.
We divide the reviewers into 5 clusters using their tf-idf vectors with Spectral Clustering. The results show that D-WBM improves on the diversity of all 1184 papers with EG of $4.28$ as shown in Table~\ref{tab:pred2}. The larger EG in this dataset compared to UIUC is because EG decreases as we reduce the upper bound, so the net effect observed in UIUC dataset is also due to choice of bounds. Here, we removed one factor and noticed much larger entropy gain for little loss of efficiency ($6\%$).

\subsection{Effect of Bounds and Problem Size}

So far, we have set the cardinality bounds without discussing their effect on the matching results. One would expect that as bounds become tighter, the utility of WBM and D-WBM will suffer due to lesser search space. However, the question we answer here is how it affects the relative performance of the two methods as measured by \PoD{} and EG.

More specifically, we study $R^-$, as the number of edges in the matching is determined by it. We use UIUC dataset discussed before, where each reviewer must review at least one paper and we cluster the reviewers into 10 groups. Figure \ref{fig:ldp} shows that the \PoD{} is consistently high irrespective of the number of reviewers matched to every paper. The \PoD{} initially decreases as more clusters contribute to diversity but then slowly increases to $1$ as the problem becomes more and more constrained. Obviously, when $R^- = M$, there is only one matching possible and both WBM and D-WBM have the same utility. EG in general increases when $R^-$ is less than the number of clusters as new clusters can contribute to diversity. Among other bounds, setting $R^+$ to any value has no effect on optimization. Increasing $L^+$ allows WBM to overuse few good nodes, who might have low edge weights with everyone. Hence, WBM's entropy suffers significantly.

\begin{figure}
\centering
\includegraphics[width=0.75\columnwidth]{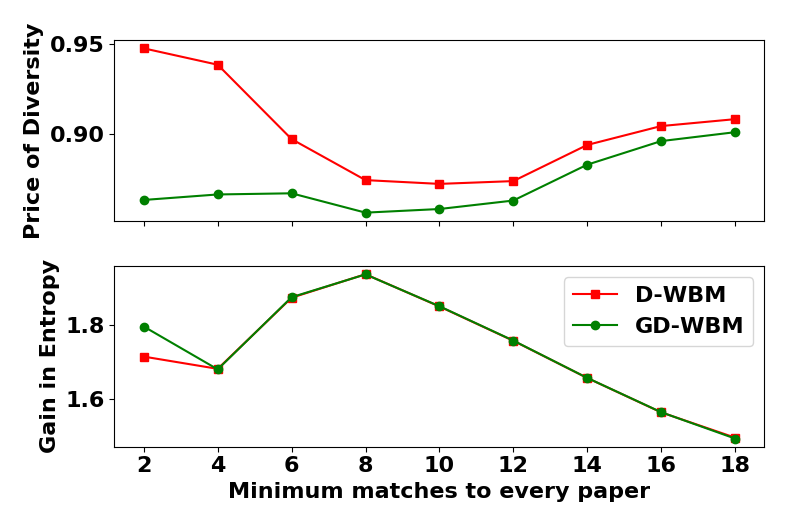}
\vspace{-0.2cm}
\caption{Change of \PoD{} and EG with increasing $R^-$.}
\label{fig:ldp}
\end{figure}

\begin{figure}
\centering
\includegraphics[width=0.85\columnwidth]{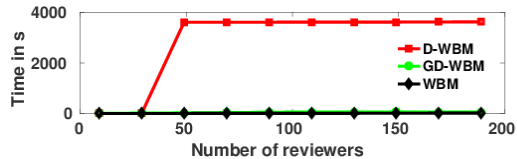}
\vspace{-0.2cm}
\caption{Runtime comparison as problem size increases}
\label{fig:gtime}
\end{figure}

Finally, we discuss the effect of problem size on the performance of WBM, D-WBM, and GD-WBM. We use the UIUC dataset with $R^- =3, L^-=1$ and increase the number of reviewers available to review the papers. 
Figure \ref{fig:gtime} shows the relative time performance of the three methods. We can see that WBM and GD-WBM take much less time than D-WBM's MIQP solver. The latter time is capped at one hour and the best solution at that point is used for analysis. 

\section{Conclusions \& Future Research}\label{sec:conclusions}

In this paper, we presented a quantitative approach to balancing diversity and efficiency in a generalization of bipartite matching where agents on one side of the market can be matched to sets of agents or items on the other.  We propose a quadratic programming-based approach to solving a supermodular minimization problem that balances diversity and total weight of the solution.  The general problem is NP-hard, so we proposed a scalable greedy algorithm with theoretical performance bounds.  We proposed the price of diversity (\PoD{}), which measures efficiency loss due to enforcing diversity, and gave worst-case theoretical bounds on that metric.  Finally, we validated our methods on three real-world datasets, and showed that the price of diversity is quite good in practice.

Future research will focus on diverse matching for online problems, where edges arrive sequentially and on scaling diverse matching to larger datasets. Another area of work can be diverse matching with diversity of a set defined using Determinantal Point Process (DPP) kernels \cite{kulesza2012determinantal}, which do not require explicit clustering. Lastly, the trade-off between diversity and efficiency can be further explored using an objective function which combines efficiency maximizing WBM and entropy maximizing D-WBM.
 
\let\oldthebibliography=\thebibliography
  \let\endoldthebibliography=\endthebibliography
  \renewenvironment{thebibliography}[1]{%
    \begin{oldthebibliography}{#1}%
      \setlength{\parskip}{0ex}%
      \setlength{\itemsep}{0ex}%
  }%
  {%
    \end{oldthebibliography}%
  }
\bibliographystyle{named}
\bibliography{refs,new_refs}
\end{document}